\documentclass[12pt,a4paper]{article}

\usepackage{authblk}
\usepackage{graphicx}
\usepackage{url}
\usepackage{amstext}
\usepackage{amsmath}
\usepackage{amsfonts}
\usepackage{subcaption}
\usepackage{mathtools}
\usepackage{float}
\usepackage{fancyhdr}
\usepackage{amsthm}
\usepackage{multirow}
\usepackage{cite}
\usepackage{placeins}

\newtheorem{theorem}{Theorem}[section]
\newtheorem{definition}[theorem]{Definition}
\newtheorem{corollary}[theorem]{Corollary}
\newtheorem{lemma}[theorem]{Lemma}

\newtheorem{property}{Property}

\usepackage[margin=1.0in]{geometry}
\title{The structure of covtree: searching for manifestly covariant causal set dynamics}

\author{Stav Zalel}
\affil{Blackett Laboratory, Imperial College London, SW7 2AZ, U.K.}

\usepackage[hidelinks]{hyperref}

\newcommand{\lc}[1]{\tilde{#1}}

\def\twoch{\,\,\begin{picture}(0,1) 
\thicklines
\multiput(0,0)(0,10){2}{\circle*{2}}
\put(0,0){\line(0,1){10}}
\end{picture}\,\,}

\def\exvee{\,\,\begin{picture}(0,1) 
\thicklines
\multiput(0,0)(0,10){3}{\circle*{2}}
\put(0,0){\line(0,1){10}}
\put(0,10){\line(0,1){10}}
\put(0,0){\line(-1,2){5}}
\put(-5,10){\circle*{2}}
\end{picture}\,\,}

\begin{document}

\maketitle

\abstract{\textit{Covtree} - a partial order on certain sets of finite, unlabeled causal sets - is a manifestly covariant framework for causal set dynamics. Here, as a first step in picking out a class of physically well-motivated covtree dynamics, we study the structure of covtree and the relationship between its paths and their corresponding infinite unlabeled causal sets. We identify the paths which correspond to posts and breaks, prove that covtree has a self-similar structure, and write down a transformation between covtree dynamics akin to the cosmic renormalisation of Rideout and Sorkin's Classical Sequential Growth models. We identify the paths which correspond to causal sets which have a unique natural labeling, thereby solving for the class of dynamics which give rise to these causal sets with unit probability.}

\newpage
\tableofcontents

\newpage
\section{Introduction}

What is the role of general covariance in quantum gravity? In causal set theory (CST), where the quantum dynamics is still unknown, clues to this question come from studying the role of general covariance at the level of classical stochastic toy models. In CST, general covariance takes the form of label-independence: two causal sets are physically equivalent if they are order-isomorphic. This has a clear consequence for the observables in CST, namely that they cannot pertain to any labeling of the causal set or to the identity of the causal set elements.

But could general covariance also have direct consequences for the dynamics? This is indeed the case in the Classical Sequential Growth (CSG) models \cite{Rideout:1999ub,Varadarajan:2005gg} which satisfy the so-called Discrete General Covariance (DGC) condition. In these models a causal set (causet) grows probabilistically through a sequential birth of elements, and the order of births induces a labeling of the elements by the natural numbers. The DGC condition constrains the dynamics such that, if $\lc{C}_n$ and $\lc{C}_n'$ are order-isomorphic causets of cardinality $n$, then $\lc{C}_n$ and $\lc{C}_n'$ are equally likely to have been grown after the birth of the first $n$ elements.

In CSG models general covariance enters via the choice of observables and via the choice of dynamics, but general covariance is implicit because the history space is the space of infinite labeled causets\footnote{See section 2 for terminology.} and so each realisation is a labeled causet (and therefore not covariant). In contrast, in a manifestly covariant approach the history space would be the space of infinite orders, denoted by $\Omega$. In this case general covariance need no longer be imposed at the level of the observables because each realisation is covariant, but whether general covariance plays an additional role at the level of the dynamics is still unknown. Additionally, it is hoped that a manifestly covariant approach would be more amenable to quantisation.

A first proposal for such a manifestly covariant framework for classical causet dynamics is \textit{covtree} \cite{Dowker:2019qiz}. Covtree is a directed tree, each of whose nodes $\Gamma_n$ at level $n$ is a set of $n$-orders. For each infinite directed path from the origin, $\mathcal{P}=\{\Gamma_1,\Gamma_2,...\}$, there exists at least one infinite order $C$ whose set of $n$-stems, for every $n>0$, is $\Gamma_n\in\mathcal{P}$. We call $C$ a \textit{certificate} of $\mathcal{P}$. Via the correspondence between paths and their certificates, any set of Markovian transition probabilities on covtree defines a causet dynamics with the space of infinite orders, $\Omega$, acting as the history space (\textit{i.e.} a manifestly covariant dynamics). 

But not every set of Markovian transition probabilities on covtree defines a physically interesting dynamics. Identifying a subset of interesting dynamics is the motivation for this current work. One challenge lies in the translation of physically desirable conditions (\textit{e.g.} that manifold-like\footnote{We say an order $C$ is manifold-like if a representative of $C$ can be faithfully embedded into a four-dimensional Lorentzian manifold.} orders are preferred by the dynamics) into conditions on covtree transition probabilities. Doing so requires an understanding of the relationship between paths and their certificates (\textit{e.g.} which paths have manifold-like certificates). Closely related challenges include formulating a causality condition on covtree and understanding what additional constraints general covariance may impose on the transition probabilities (cf. the DGC condition in CSG models).

In addition to the relationship between paths and their certificates, an understanding of the structure of covtree is also important for constraining the dynamics. For example, any dynamics should satisfy the Markov sum rule: the sum of the transition probabilities from a node $\Gamma_n$ to each of its children must equal 1. But with no knowledge of the number of children or the relation they bear to $\Gamma_n$, this constraint is intractable. (In contrast, in the case of the CSG models, enough structural information is known to solve the Markov sum rule.)

In addressing these challenges, one might be tempted to construct covtree explicitly. Indeed, the first three levels of covtree are given in \cite{Dowker:2019qiz}, but brute force methods come up short in going to higher levels as the number of candidate nodes at level $n$ increases rapidly as $2^{|\Omega(n)|}-1$, where $|\Omega(3)|=5, |\Omega(5)|=63$ and $|\Omega(16)|=4483130665195087$ \cite{oeis}. In this work we make progress by focusing on structural properties which are independent of level.

The rest of this paper is structured as follows. Section 2 is dedicated to a presentation of terminology and notation. In section 3, led by the ideas of causal set cosmology, we identify the covtree paths whose certificates contain posts and breaks, prove that covtree has a self-similar structure and find the covariant analogue of the cosmic renormalisation transformation of CSG models. In section 4 we present additional structural features of covtree, as well as a toy example of how this structure can be used to constrain the dynamics. We conclude with a discussion in section 5.

\section{Terminology and notation}

Here we present a brief review of terminology and notation used in the paper. For further discussion and examples, we refer the reader to \cite{Dowker:2019qiz} and references therein. For a recent review of causal set theory, see \cite{Surya:2019ndm}.

\subsection{Labeled causets, orders and stems}

For any natural number $n$, let $[n]$ denote the set $\{0,1,...,n\}$.

A \textbf{labeled causet} is a locally finite partial order on ground set $[n]$ or $\mathbb{N}$ such that $x\prec y \implies x<y$. A labeled causet of cardinality $n$ is called an \textbf{$n$-causet}. We denote labeled causets and their subcausets by capital Roman letters with a tilde, \textit{e.g.} $\lc{C}$. 

If $x\prec y$ in $\lc{C}$ we say that $y$ is a \textbf{descendant} of $x$ or that $y$ is \textbf{above} $x$. If $x\prec y$ and there is no $z\in\lc{C}$ such that $x\prec z \prec y$ we say that $y$ is a \textbf{direct descendant} of $x$ or that $y$ is \textbf{directly above} $x$ or that $y$ is a \textbf{child} of $x$.  The \textbf{valency} of $x$ is the number of direct descendants of $x$.

An element $x\in \lc{C}$ is in \textbf{level $L$} in $\lc{C}$ if the longest chain of which $x$ is the maximal
element has cardinality $L$, \textit{e.g.} level 1 comprises the minimal elements. For any $x\in \lc{C}$, $L(x)$ is an integer which denotes the level of $x$, \textit{e.g.} $L(x)=1$ if $x$ is minimal.

If causets $\lc{C}$ and $\lc{D}$ are order-isomorphic we write $\lc{C}\cong\lc{D}$.

An \textbf{order} is an order-isomorphism class of labeled causets. An \textbf{$n$-order} is an order-isomorphism class of $n$-causets. We denote orders by capital Roman letters without a tilde, \textit{e.g.} $C$. The \textbf{cardinality of an $n$-order} is defined to be $n$. We denote the cardinality of an order $C$ by $|C|$.

We often (but not always) use a subscript to denote the cardinality of an $n$-causet or $n$-order, \textit{e.g.} $\lc{C}_n$ or $C_n$. 

A \textbf{stem} in a labeled causet $\lc{C}$ is a finite subcauset $\lc{S}$ in $\lc{C}$ such that if $y\in \lc{S}$ and $x\prec y$ in $\lc{C}$ then $x \in \lc{S}$. We say a finite order, $S$, is a stem in order $C$ if there exists a representative of $S$ which is a stem in a representative of $C$. We say a finite order, $S$, is a stem in labeled causet $\lc{C}$ if the order $S$ is a stem in the order $[\lc{C}]$. So the meaning of stem depends on context. If a stem has cardinality $n$ we say it is an \textbf{$n$-stem}.

An infinite order $C$ is a \textbf{rogue} if there exists an infinite order $D$ such that $D \neq C$ and the two orders have the same stems.

\subsection{Labeled poscau and CSG dynamics}

Let $\tilde{\Omega}(\mathbb{N})$ and $\tilde{\Omega}$ denote the set of all finite and infinite labeled causets, respectively.

\textbf{Labeled poscau} is the partial order $(\tilde{\Omega}(\mathbb{N}), \prec)$, where $\lc{S}\prec \lc{R}$ if and only if $\lc{S}$ is a stem in $\lc{R}$.\footnote{We use the symbol $\prec$ to denote the relation for several different partial orders in this work. The meaning of $\prec$ in each case is to be inferred from the context.}

A \textbf{poscau dynamics} is a complete set of Markovian transition probabilities on labeled poscau. We denote the probability of transition from $\lc{C}_n$ to one of its children $\lc{C}_{n+1}$ by $\mathbb{P}(\lc{C}_n\rightarrow\lc{C}_{n+1})$. We denote the probability of a directed random walk to pass through $\lc{C}_n$ by $\mathbb{P}(\lc{C}_n)$.  

\textbf{Classical Sequential Growth (CSG) models} are a family of poscau dynamics which satisfy the so-called Bell Causality and Discrete General Covariance conditions of \cite{Rideout:1999ub}. Each model in the family is specified by an infinite set of real positive coupling constants, $\{t_k\}_{k\in\mathbb{N}}$, with $t_0>0$. The CSG transition probabilities take the form:
\begin{align}\label{transprob}
\mathbb{P}(\tilde{C}_n \rightarrow \tilde{C}_{n+1}) &= \frac{\lambda(\varpi, m)}{\lambda(n, 0)}\,,
\end{align}
where $\varpi$ and $m$ are positive integers which depend on $\tilde{C}_n$ and $\tilde{C}_{n+1}$, and  
\begin{equation} 
\lambda(k, p) := \sum_{i = 0}^{k - p} \binom{k-p}{i} t_{p+i}.\end{equation}   

\textbf{Originary CSG models} are a family of poscau dynamics which differ from CSG models only by the requirement that $t_0=0$.

Each poscau dynamics is equivalent to a measure space $(\lc{\Omega},\lc{\mathcal{R}},\lc{\mu})$, where $\lc{\Omega}$ is the set of infinite labeled causets, $\lc{\mathcal{R}}$ is the sigma-algebra generated by the cylinder sets, 
\begin{equation}cyl(\lc{C}_n):=\{\lc{C}\in\tilde{\Omega}|\lc{C}_n \text{ is a stem in } \lc{C}\},\end{equation}
and the measure $\lc{\mu}$ is defined via $\mu(cyl(\lc{C}_n))=\mathbb{P}(\lc{C}_n)$ for every $\lc{C}_n$.
The covariant sigma-algebra, $\mathcal{R}$, is a sub-algebra of $\lc{\mathcal{R}}$ defined by 
\begin{equation}\mathcal{R}:=\{\mathcal{E}\in \lc{\mathcal{R}}| \lc{C}\in\mathcal{E} \text{ and }  \lc{C}\cong\lc{D} \implies \lc{D}\in\mathcal{E}\}.\end{equation}

\subsection{Covtree}

Let $\Omega$ denote the set of infinite orders. Let $\Omega(n)$ denote the set of all $n$-orders for some $n\in\mathbb{N}^+$, and let $\Gamma_n$ denote a non-empty subset of $\Omega(n)$, \textit{i.e.} a set of $n$-orders.

An order $C$ is a \textbf{certificate} of $\Gamma_n$ if $\Gamma_n$ is the set of all $n$-stems in $C$. If $C$ is a certificate of $\Gamma_n$ and there exists no stem $D\not=C$ in $C$ which is a certificate of $\Gamma_n$, then we say that $C$ is a \textbf{minimal certificate} of $\Gamma_n$. A labeled causet $\lc{C}$ is a \textbf{labeled certificate} of $\Gamma_n$ if it is a representative of a certificate $C$ of $\Gamma_n$.

For any $n$ and any $\Gamma_n$, the map ${{\mathcal O}}_{-}$ takes $\Gamma_n$ to the set of $(n-1)$-stems of elements of $\Gamma_n$:
\begin{equation} 
{\mathcal O}_{-}(\Gamma_n):=\{B \in \Omega(n-1)\ | \  \exists \ A\in \Gamma_n\ \mathrm{ s.t. }\ B \text{ is a stem in } A \}
\,. \label{o_minus_defn}
\end{equation} 

Let $\Lambda$ denote the collection of sets of $n$-orders, for all $n$, which have certificates:
\begin{align}
\Lambda := \bigcup\limits_{n\in\mathbb{N}^+}\{\Gamma_n \subseteq \Omega(n)| \exists \text{ a certificate for } \Gamma_n\}\,. 
\end{align}

\textbf{Covtree} is the partial order $(\Lambda, \prec)$, where $\Gamma_n\prec\Gamma_{m}$ if and only if $n<m$ and ${\mathcal {O}_{-}}^{m-n}(\Gamma_m)=\Gamma_n$.

If  $\Gamma_n\in\Lambda$, we say that $\Gamma_n$ is a \textbf{node} in covtree. If $\Gamma_n$ is a node in covtree and $\Gamma_n$ contains a single $n$-order, we say that  $\Gamma_n$ is a \textbf{singleton}.  If $\Gamma_n$ is a node in covtree and $\Gamma_n$ contains exactly two $n$-orders, we say that  $\Gamma_n$ is a \textbf{doublet}.

Let $\mathcal{P}=\{\Gamma_1,\Gamma_2,\Gamma_3,...\}$ be a path in covtree. An infinite order $C$ is a \textbf{certificate of path $\mathcal{P}$} if $\Gamma_n\in \mathcal{P}$ is the set of $n$-stems in $C$, for every $n$. Every path has at least one certificate, and every infinite order is a certificate of exactly one path.
 
A \textbf{covtree dynamics} is a complete set of Markovian transition probabilities on covtree. We denote the probability of transition from $\Gamma_n$ to one of its children $\Gamma_{n+1}$ by $\mathbb{P}(\Gamma_n\rightarrow\Gamma_{n+1})$. We denote the probability of a directed random walk to pass through $\Gamma_n$ by $\mathbb{P}(\Gamma_n)$. We denote a covtree dynamics by $\{\mathbb{P}\}$.
 
The \textbf{certificate set}, $cert(\Gamma_n)$, of some node $\Gamma_n$ is the set of infinite orders which are certificates of $cert(\Gamma_n)$:
\begin{equation}\label{eq25072002}cert(\Gamma_n)=\{C\in\Omega|C \text{ is a certificate of } \Gamma_n\}.\end{equation}

Each covtree dynamics, $\{\mathbb{P}\}$, is equivalent to a measure space $(\Omega,\mathcal{R}(\mathcal{S}),\mu)$, where $\Omega$ is the set of infinite orders, $\mathcal{R}(\mathcal{S})$ is the sigma-algebra generated by the certificate sets\footnote{In the literature (\textit{e.g.} \cite{Brightwell:2002vw}) $\mathcal{R}(\mathcal{S})$ denotes the sigma-algebra generated by the \textit{stem sets}. It is a result of \cite{Dowker:2019qiz} that the sigma-algebras generated by the stem sets and the certificate sets are equal.} and $\mu$ is a measure defined via $\mu(cert(\Gamma_n))=\mathbb{P}(\Gamma_n)$ for all $\Gamma_n$.

Equivalently, we can conceive of each certificate set as a set of labeled causets:
\begin{equation}cert(\Gamma_n):=\{\lc{C}\in\lc{\Omega}|\lc{C} \text{ is a labeled certificate of } \Gamma_n\}.\end{equation}
In this formulation, $\mathcal{R}(\mathcal{S})$ is a set of subsets of $\lc{\Omega}$. In fact, $\mathcal{R}(\mathcal{S})\subset\mathcal{R}\subset\lc{\mathcal{R}}$ and therefore every poscau dynamics induces a covtree dynamics via a restriction of the measure $\lc{\mu}$ from $\lc{\mathcal{R}}$ to $\mathcal{R}(\mathcal{S})$ . We say that a covtree dynamics is a CSG dynamics if it is the restriction of a CSG dynamics.

\section{Covtree and causal set cosmology}\label{sec_renorm}

In the heuristic causal set cosmology paradigm proposed in \cite{Sorkin:1998hi}, the cosmos emerges from the quantum gravity era sufficiently flat, homogeneous and isotropic to explain present-day observations (without the need for a period of inflation). Within the context of CSG models, the fine-tuning problem is that of choosing a CSG dynamics which displays this behavior almost surely. The need for fine-tuning is overcome by a ``cosmic renormalisation'' associated with cycles of expansion and collapse, punctuated by Big-Crunch--Big-Bang singularities.

At least heuristically, posts and breaks are the causal set structures which underlie Big-Crunch--Big-Bang singularities  \cite{Martin:2000js,Dowker:2017zqj}. Let $\lc{C}$ be a labeled causet. A \textbf{post} is an element $x\in\lc{C}$ which is related to every other element in $\lc{C}$. A \textbf{break} in $\lc{C}$ is an ordered pair, $(\lc{A},\lc{B})$, of nonempty subsets of $\lc{C}$ such that
\begin{itemize}
\item[(i)] $a\in \lc{A}, b\in \lc{B} \implies a\prec b$, and
\item[(ii)] $\{\lc{A},\lc{B}\}$ is a partition of $\lc{C}$.
\end{itemize}
We call $\lc{A}$ and $\lc{B}$ the past and future of the break, respectively. If $\lc{C}$ contains a break with past $\lc{A}$ we say that $\lc{C}$ contains an $\lc{A}$-break. If $\lc{C}$ contains a post $x$ with $past(x)=\lc{A}$ we say that $\lc{C}$ contains an $\lc{A}$-post\footnote{We are using the non-inclusive past convention, so $x\notin past(x)$.}. An illustration of a post and a break is shown in figure \ref{pb}.

 \begin{figure}[h!]
  \centering
	\includegraphics[width=0.5\textwidth]{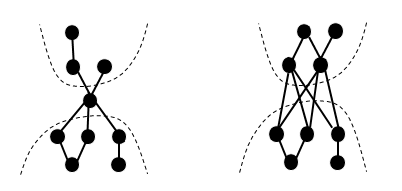}
	\caption{Illustration of a post (left) and a break (right). On the left, the post lies between the dashed lines. On the right, the dashed lines illustrate the partition between the past and the future.}
\label{pb}
\end{figure}

In causal set theory nothing can be smaller than a single spacetime atom, therefore our intuition points us towards a post as the underlying structure of a singularity caused by collapse. Meanwhile, the break is a generalisation of the post which retains the key feature we require of Big-Crunch--Big-Bang singularities: the partitioning of the set into two, a past and a future.

Associated with each post and break is a cosmic renormalisation transformation which acts on the CSG couplings. The renormalised couplings define an effective dynamics which governs the growth of the future, echoing proposals by Wheeler, Smolin and others that the parameters of nature are modified as the universe is ``squeezed through'' a singularity \cite{Rees:1974,wheelerint,universeevolve,status}.

In this cosmological paradigm the fine-tuning problem is resolved by an evolutionary mechanism. It rests on the hypothesis that there exist stationary points of the renormalisation transformation which give rise to the desired cosmological features, and that the basin of attraction of these stationary points is large and contains an abundance of dynamics each of which gives rise to an infinite sequence of Big-Crunch--Big-Bang singularities with unit probability \cite{Martin:2000js,Ash:2005za, Ash:2002un, Dou:1999fw, Brightwell:2016}. Given this, no fine-tuning is required for our universe to be governed by a dynamics in the basin of attraction which gives rise to an infinite sequence of singularities. At each singularity the couplings undergo a renormalisation, and in this way a flow towards the stationary point is generated in the space of couplings. It is then only a matter of time until our universe displays the desired behaviour.

This narrative acts as guidance as to which covtree dynamics we should be seeking, namely:
\begin{itemize}
\item[(a)] dynamics in which an infinite sequence of posts or breaks happens almost surely,
\item[(b)] dynamics which are stationary points of the renormalisation transformation,
\item[(c)] dynamics which flow to a stationary point under the renormalisation transformation.
\end{itemize}

We begin section \ref{subsec_131201} by recasting the definitions of post and break in covariant form\footnote{The definitions of post and break given above are not covariant because they pertain to labeled causets, not orders.}. We identify which covtree paths correspond to orders with posts and breaks and write down the defining feature of covtree walks which belong to family $(a)$. In section \ref{subsec_151202} we show that covtree has a self-similar structure. In section \ref{subsec_151201} we use covtree's self-similarity to solve for the covtree walks which belong to family (b). We conclude with a discussion of open questions, including a proposal for a causality condition for covtree dynamics.

\subsection{Certificates with posts and breaks}\label{subsec_131201}

In solving for the covtree walks which belong to family $(a)$, a question arises: which paths correspond to orders with posts and breaks? To pose this question more precisely, let us extend the definitions of posts and breaks to orders. Let $\lc{C}$ and $\lc{A}$ be representatives of orders $C$ and $A$, respectively. We say that an order $C$ contains a break (post) with past $A$ if $\lc{C}$ contains a break (post) with past $\lc{A}$. If order $C$ contains a break (post) with past $A$ we say that $C$ contains an $A$-break ($A$-post). Our question then becomes: which paths have certificates with posts and breaks?

To answer this question we introduce the concept of the \textbf{covering causet}. If $\lc{C}$ is a labeled causet of cardinality $n$ then its covering causet is the labeled causet of cardinality $n+1$ which is formed by putting the element $n$ above every element of $\lc{C}$, and we denote it by putting a hat on: $\widehat{\lc{C}}$. If $\lc{C}$ and $\widehat{\lc{C}}$ are representatives of orders $C$ and $\widehat{{C}}$, respectively, then we say that $\widehat{C}$ is the \textbf{covering order} of $C$. An example is shown in figure \ref{cover}.

We will show that:
\begin{theorem}\label{lemma131201}
Let order $C$ be a certificate of a path $\mathcal{P}=\{\Gamma_1, \Gamma_2, ...\}$. \begin{enumerate}\item $C$ contains an $A$-break if and only if $\{\widehat{A}\}$ is a node in $\mathcal{P}$. \item $C$ contains an $A$-post if and only if $\{\Hat{\Hat{A}}\}$ is a node in $\mathcal{P}$ (where $\Hat{\Hat{A}}$ is the covering order of the covering order of $A$). \end{enumerate} \end{theorem}
 
To prove theorem \ref{lemma131201} we will need the following lemma about labeled causets which contain breaks:
\begin{lemma}\label{lemma111201}
Let $\lc{C}$ and $\lc{A}$ be labeled causets. The following statements are equivalent:
\begin{itemize}
\item[(i)] $\lc{C}$ contains an $\lc{A}$-break,
\item[(ii)] every $(|\lc{A}|+1)$-stem in $\lc{C}$ is isomorphic to $\widehat{\lc{A}}$,
\item[(iii)] $\lc{A}$ is the unique $|\lc{A}|$-stem in $\lc{C}$.
\end{itemize}
\end{lemma}

\begin{proof}
Let $\lc{A}$ be a stem in $\lc{C}$. Let $x$ denote a minimal element in $\lc{C}\setminus\lc{A}$. Let $a$ denote an element in $\lc{A}$.
\begin{itemize}
\item[$(i)\implies(ii)$] It follows from the definition of an $\lc{A}$-break that every $(|\lc{A}|+1)$-stem in $\lc{C}$ is of the form $\lc{A}\cup\{x\}$ and that each such stem is isomorphic to $\widehat{\lc{A}}$.
\item[$(ii)\implies(iii)$] Suppose for contradiction that $\lc{D}$ is an $|\lc{A}|$-stem in $\lc{C}$, and $\lc{D}\not=\lc{A}$. Let $y$ be minimal in $\lc{D}\setminus \lc{A}$. Then $\lc{A}\cup\{y\}$ is an $(|\lc{A}|+1)$-stem in $\lc{C}$. By assumption $(ii)$, $y\succ a$ for all $a$ in $\lc{A}$, and therefore (by definition of stem) $\lc{A}\cup\{y\}\subseteq \lc{D}$ which in turn implies that $|\lc{D}|>|\lc{A}|$. Contradiction. 
\item[$(iii)\implies(i)$] By assumption $(iii)$, $x\succ a$ for all $x$ and $a$, and hence (by definition of break) $\lc{C}$ contains an $\lc{A}$-break.
\end{itemize}
\end{proof}

The following covariant statement is a corollary:
\begin{corollary}\label{cor_161201}
An order $C$ contains an $A$-break if and only if $\widehat{A}$ is its unique $(|A|+1)$-stem.
\end{corollary}

We can now prove theorem \ref{lemma131201}:
\begin{proof}[Proof of theorem \ref{lemma131201}.] Let $C$ be a certificate of path $\mathcal{P}=\{\Gamma_1, \Gamma_2, ...\}$. Recall that, by definition, the set of $n$-stems of $C$ is the node $\Gamma_n$ in $\mathcal{P}$.

To prove part 1, suppose $C$ contains an $A$-break. Then by corollary \ref{cor_161201}, the set of $(|A|+1)$-stems of $C$ is $\{\widehat{A}\}$. By definition of certificate, $\{\widehat{A}\}$ is in $\mathcal{P}$. Now suppose $\{\widehat{A}\}$ is in $\mathcal{P}$. Then by definition of certificate, $\{\widehat{A}\}$ is the set of $(|A|+1)$-stems of $C$. By corollary \ref{cor_161201}, $C$ contains an $A$-break.

Part 2 follows from the fact that $C$ contains an $A$-post if and only if $C$ contains an $\widehat{A}$-break. 
\end{proof}

We have successfully identified which paths correspond to posts and breaks, and we can now characterise the covtree walks which belong to family (a):
\begin{itemize}
\item[] A covtree dynamics in which an infinite sequence of breaks happens with unit probability is one which, with unit probability, passes through infinitely many nodes of the form $\{\widehat{A}\}$, \textit{i.e.} singletons which contain a covering order. 
\item[] A covtree dynamics in which an infinite sequence of posts happens with unit probability is one which, with unit probability, passes through  infinitely many nodes of the form $\{\widehat{\widehat{A}}\}$, \textit{i.e.} singeltons which contain a covering order of a covering order.
\end{itemize}
An illusration is shown in figure \ref{fig_15072001}.
 \begin{figure}[p]
	\centering
	\includegraphics[width=0.4\textwidth]{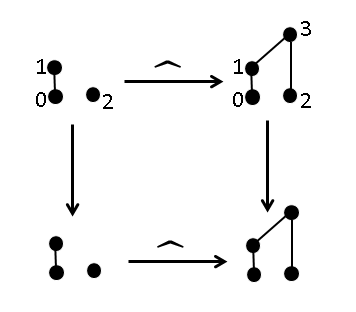}
	\caption{The horizonatal arrows illustrate the relationship between a causet (order) and its covering causet (order), as indicated by the hat. The vertical arrows illustrate the relationship between an order and its representative.}
	\label{cover}
\end{figure}

\begin{figure}[p]
	\centering
	\includegraphics[width=0.15\textwidth]{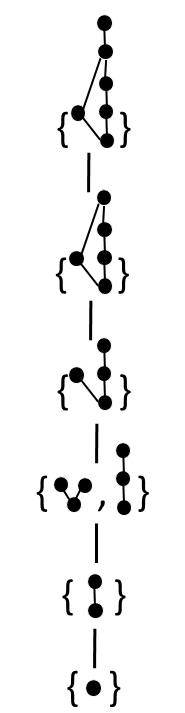}
	\caption{Illustration of a path whose certificate contains breaks and posts. The first six nodes of a covtree path $\mathcal{P}$ are shown. The nodes at level 2, 5 and 6 are singletons which contain a covering order. A certificate $C$ of $\mathcal{P}$ contains a break whose past is the 1-order (corresponding to the node at level 2) and a post whose past is the 4-order \protect\exvee (corresponding to the node at level 6). That $C$ contains this post implies that $C$ also contains a break whose past is the 4-order \protect\exvee (corresponding to the node at level 5).}
	\label{fig_15072001}
\end{figure}

\subsection{Covtree self-similarity}\label{subsec_151202}

To identify covtree dynamics which fall into families (b) and (c) we must first understand how the renormalisation transformation is manifest on covtree. This turns out to be inextricably linked to covtree's self-similar structure. In this section we identify this self-similarity.

Let us begin by defining what we mean by a self-similar structure of a partial order. Let $\Pi$ and $\Psi$ be partial orders. We say that $\Psi$ contains a copy of $\Pi$ if there exists a convex sub-order $\Pi'\subseteq \Psi$ which is order-isomorphic to $\Pi$. If $\Psi$ contains infinitely many copies of itself we say that $\Psi$ is self-similar.

Let us denote the covtree partial order by $\Lambda$. For any finite order $A$, let $\Lambda_A\subset \Lambda$ be the convex sub-order of covtree which contains the singleton $\{\widehat{A}\}$ and everything above it. We will show that:
\begin{lemma}\label{cor_29111901}  For any finite order $A$, $\Lambda_A$ is a copy of covtree. \end{lemma}
The following theorem is a corollary:
\begin{theorem}\label{theorem_05082001}
Covtree is self-similar.
\end{theorem}

We will need the following definition:
\begin{definition}\label{defn_161201} Given a finite order $A$ and a set of orders $\Upsilon$, the map $\mathcal{G}_A$ takes $\Upsilon$ to $\mathcal{G}_A(\Upsilon)$, the set  of orders which contain a break with past $A$ and future $B\in\Upsilon$, i.e.
\\[-20pt]
\begin{equation*}
\mathcal{G}_A(\Upsilon):=\{C \ | \ C \text{ is an order which contains a break with past $A$ and future }B\in\Upsilon\}.\end{equation*}\end{definition}

Examples are shown in figure \ref{G_A_defn}. Note that $\Upsilon$ may contain finite orders, infinite orders, or both.
\begin{figure}[h]
	\centering
	\includegraphics[width=0.9\textwidth]{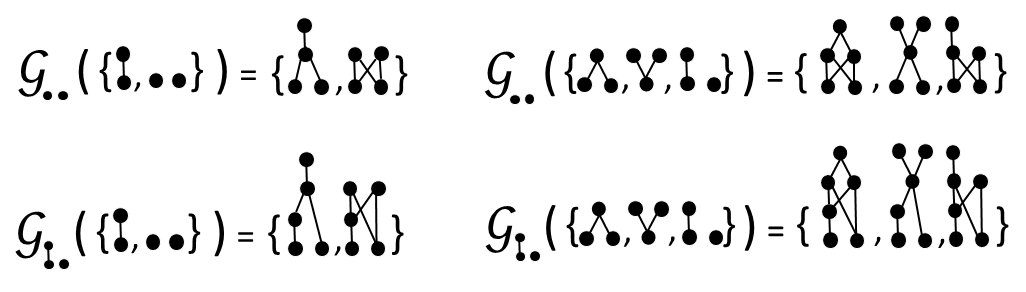}
	\caption{Illustration of the operation $\mathcal{G}_A$.}
	\label{G_A_defn}
\end{figure}

\begin{proof}[Proof of lemma \ref{cor_29111901}:] We will show that, for any finite order $A$,
\begin{enumerate}
\item[$(i)$] $\mathcal{G}_A(\Lambda)=\Lambda_A$, and
\item[$(ii)$] the map $\mathcal{G}_A:\Lambda\rightarrow\Lambda_A$ is an order-isomorphism,
\end{enumerate}
and the result follows. Note that, by definition \ref{defn_161201}, $\Lambda$ is a strict subset of the domain of $\mathcal{G}_A$. Here we use $\mathcal{G}_A$ to denote the restriction of the map to $\Lambda$. The use of $\mathcal{G}_A$ should be clear from the context.
 
To prove part $(i)$, we first show that $\mathcal{G}_A(\Lambda)\subseteq\Lambda_A$. Let $\Gamma_n\in\Lambda$ and let the $m$-order $C_m$ be a certificate of $\Gamma_n$. Let $D_k$ denote the order of cardinality $k=m+|A|$ which contains a break with past $A$ and future $C_m$. Then $D_k$ is a certificate of $\mathcal{G}_A(\Gamma_n)$, and therefore $\mathcal{G}_A(\Gamma_n)\in\Lambda$. Next, note that $\widehat{A}$ is the unique $|\widehat{A}|$-stem in every order in $\mathcal{G}_A(\Gamma_n)$, for any $\Gamma_n$. Therefore $\mathcal{G}_A(\Lambda)\subseteq\Lambda_A$.

Second, we show that $\Lambda_A\subseteq\mathcal{G}_A(\Lambda)$. Let $\Gamma_n\in\Lambda_A$ and let the $p$-order $E_p$ be a certificate of $\Gamma_n$. Necessarily, $E_p$ contains a break with past $A$ and some future $B$. Let $\Gamma_{l}$ denote the set of $l$-stems of $B$, where $l=n-|A|$. Then $\Gamma_{l}\in\Lambda$ and $\mathcal{G}_A(\Gamma_{l})=\Gamma_n$.

Therefore, $\mathcal{G}_A(\Lambda)=\Lambda_A$. 

To prove part $(ii)$, we use the commutativity of the operations ${\mathcal O}_{-}$ and $\mathcal{G}_A$ to show that $\mathcal{G}_A:\Lambda\rightarrow\Lambda_A$ is order-preserving. Suppose $\Gamma_n\prec \Gamma_{n+1}$, then by definition of covtree we have that $\Gamma_n={\mathcal O}_{-}(\Gamma_{n+1})$, and therefore
$$\mathcal{G}_A(\Gamma_n)=\mathcal{G}_A({\mathcal O}_{-}(\Gamma_{n+1}))={\mathcal O}_{-}(\mathcal{G}_A(\Gamma_{n+1})) \implies \mathcal{G}_A(\Gamma_n)\prec \mathcal{G}_A(\Gamma_{n+1}).$$ Now suppose $\mathcal{G}_A(\Gamma_n)\prec \mathcal{G}_A(\Gamma_{n+1})$. Then
\begin{equation}\begin{split} \mathcal{G}_A(\Gamma_n)={\mathcal O}_{-}(\mathcal{G}_A(\Gamma_{n+1}))=\mathcal{G}_A({\mathcal O}_{-}(\Gamma_{n+1}))\implies \Gamma_n\prec \Gamma_{n+1}.
\end{split}\end{equation}
\end{proof}

Covtree contains countably many copies of itself, each with ground set $\Lambda_A$ and root $\{\widehat{A}\}$ (where we can think of $\Lambda$ itself as $\Lambda_{\emptyset}$). An illustration is shown in figure~\ref{covtree_self_similarity}.

\begin{figure}[h]
	\centering
	\includegraphics[width=0.6\textwidth]{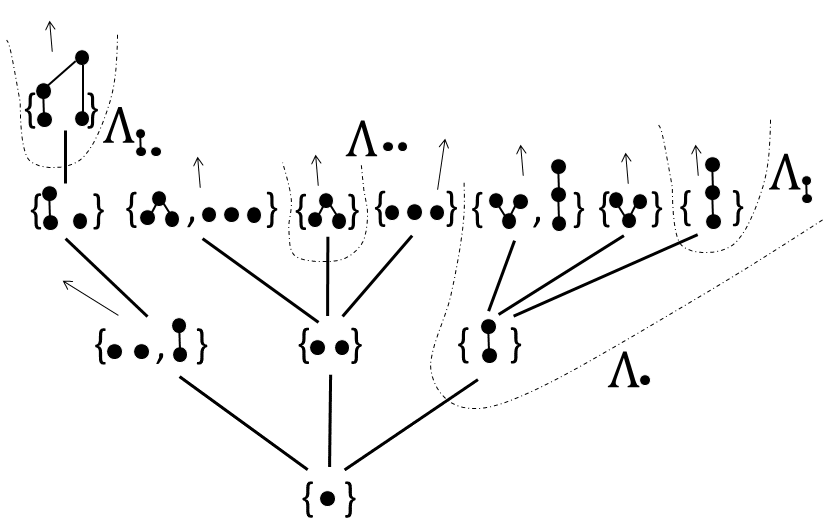}
	\caption{The self-similar structure of covtree. The figure displays the first two levels of covtree in full and selected nodes from levels 3 and 4. The arrows indicate additional nodes not shown in the figure. The dashed lines indicate where a new copy of covtree begins. The ground set of each copy is indicated next to each dashed line.}
	\label{covtree_self_similarity}
\end{figure}

\subsection{Covariant cosmic renormalisation}\label{subsec_151201}

In this section we give a brief review of cosmic renormalisation in CSG models, present its covariant counterpart and conclude with a discussion of open questions.

\paragraph{The labeled break condition:}
In CSG models, cosmic renormalisation comes about as a result of a conditioning. One conditions on an $\lc{A}$-break, where $\lc{A}$ is some labeled causet of choice. The transition probabilities which are not consistent with the break condition are set to zero, while the remaining transition probabilties are normalised to satisfy the Markov sum rule.

The break condition fixes the past but leaves the future unconstrained and dynamical. Consequently, one can conceive of the growth of the future as a new poscau walk with a new poscau dynamics -- the \textit{effective dynamics}, derived from the original dynamics via the break condition.

Remarkably, the effective dynamics is itself a CSG model and one can think of the coupling constants as undergoing a renormalisation: 
\begin{equation}\begin{split}\label{eq_181201}
&Q_{m,r}:\{t_k\}\rightarrow\{t^{(m,r)}_k\}\\
&t^{(m,r)}_0=\sum_{l=0}^{m-r}\binom{m-r}{l}t_{r+l}\\
&t^{(m,r)}_k=\sum_{l=0}^m\binom{m}{l}t_{k+l} \ \forall \ k>0,\end{split}\end{equation}
where $\{t_k^{(m,r)}\}$ is the set of renormalised couplings, $m=|\lc{A}|$ and $r$ is the number of maximal elements of $\lc{A}$. The composite label $m,r$ on the renormalisation transformation, $Q_{m,r}$, and on the renormalised couplings, $t_k^{(m,r)}$, signifies that the transformation depends on $\lc{A}$ only via these two quantities. This is an attractive feature of CSG models: the causal structure of the past is ``forgotten'' by the effective dynamics, providing still further motivation for regarding the past and the future of the break as separated. A derivation of transformation \ref{eq_181201} can be found in \cite{Dowker:2017zqj}.

A sequence of breaks corresponds to a sequence of applications of the transformation $Q_{m,r}$ with the values of $m$ and $r$ varied appropriately. In this way a flow is generated in the space of couplings.

A stationary point is a dynamics which is mapped onto itself.
The unique family of stationary points of $Q_{m,r}$ is given by $t_0=(1+t)^{-r}t^r$ and $t_k=t^k$ for $k>0$, where $t$ is any positive real number (or equivalently, $t_0=1$ and $t_k=(1+t)^rt^{k-r}$ for $k>0$). Since the stationary points depend on $r$ but not on $m$, a pair of transformations, $Q_{m,r}$ and $Q_{n,s}$, have either no stationary points in common (if $s\not=r$) or have exactly the same set of stationary points (if $s=r$).

\paragraph{The covariant break condition:}
What form does cosmic renormalisation take on covtree? First, one conditions on an $A$-break --- this is the covariant break condition. By theorem \ref{lemma131201}, the condition constrains the covtree walk to pass through the node $\{\widehat{A}\}$ but leaves the walk unconstrained thereafter. Since we do not (yet) know how to encode a covtree dynamics as a set of couplings, there are no couplings to undergo a renormalisation. Nevertheless, there is an effective covtree dynamics which governs the growth of the future, and the associated transformation acts on the transition probabilities directly. A covtree dynamics $\{\mathbb{P}\}$ and its corresponding effective dynamics $\{\mathbb{P}_A\}$ are related via the transformation $R_A$:
\begin{equation}\begin{split}\label{eq2507201}
&R_A:\{\mathbb{P}\}\rightarrow\{\mathbb{P}_A\}\\
&\mathbb{P}_A(\Gamma_n\rightarrow\Gamma_{n+1})=\mathbb{P}(\mathcal{G}_A(\Gamma_n)\rightarrow\mathcal{G}_A(\Gamma_{n+1})),
\end{split}\end{equation}
where $\mathcal{G}_A$ is the mapping introduced in definition \ref{defn_161201}.
 For a generic covtree dynamics, the functional relationship between $\{\mathbb{P}\}$ and $\{\mathbb{P}_A\}$ can depend on any feature of $A$, and this is signified by the label $A$ on the transformation, $R_A$, and on the effective transition probabilities, $\{\mathbb{P}_A\}$.

Let us sketch the derivation of transformation \ref{eq2507201}. Recall that every covtree dynamics $\{\mathbb{P}\}$ is equivalent to a measure space $(\Omega,\mathcal{R}(S),\mu)$, where $\mu(cert(\Gamma_m))=\mathbb{P}(\Gamma_m)$. Let  $(\Omega,\mathcal{R}(S),\mu_A)$ denote the measure space equivalent to the effective dynamics $\{\mathbb{P}_A\}$. Then the effective dynamics is defined by:
\begin{equation}\label{eq2507202}
\mu_A(\mathcal{E}):=\frac{\mu(\mathcal{G}_A(\mathcal{E}))}{\mu(cert(\{\widehat{A}\}))} \ \forall \ \mathcal{E}\in\mathcal{R}(S).
\end{equation}
Now, choose $\mathcal{E}=cert(\Gamma_n)$ for some node $\Gamma_n$. This sets $\mathcal{G}_A(\mathcal{E})=\mathcal{G}_A(cert(\Gamma_n))=cert(\mathcal{G}_A(\Gamma_n))$. Then, use the relations $\mu(cert(\Gamma_m))=\mathbb{P}(\Gamma_m)$ and $\mu_A(cert(\Gamma_m))=\mathbb{P}_A(\Gamma_m)$ to rewrite condition \ref{eq2507202} in terms of probabilities. Finally, we can use induction to reach transformation \ref{eq2507201}.

The transformation $R_A$ acts directly on transition probabilities, not on a set of couplings. To emphasis this, we call $R_A$ a \textit{similarity} transformation rather than a \textit{renormalisation} transformation. If (as one hopes) in future we are able to characterise a covtree dynamics by a set of couplings, it may be possible to write the similarity transformation as a renormalisation transformation which acts on the couplings directly.

Similarly, we reserve the term \textit{stationary point} for a set of couplings which is mapped onto itself by the renormalisation transformation. If a covtree dynamics is mapped onto itself by a similarity transformation $R_A$, we say that it is \textit{self-similar} with respect to $R_A$. A covtree dynamics $\{\mathbb{P}\}$ is self-similar with respect to $R_A$ if and only if it satisfies the condition
\begin{equation}\label{eq_181207}\mathbb{P}(\Gamma_n\rightarrow\Gamma_{n+1})=\mathbb{P}(\mathcal{G}_A(\Gamma_n)\rightarrow\mathcal{G}_A(\Gamma_{n+1}))\end{equation} for every $n$ and every transition $\Gamma_n\rightarrow\Gamma_{n+1}$.
Constructing a self-similar dynamics is simple: assign any set of transition probabilities to the transitions which lie outside $\Lambda_A$, and then use equality \ref{eq_181207} to set the transition probabilities in $\Lambda_A$.  

It is possible to use this procedure to fix the transition probabilities in $\Lambda_A$ for every $A$ simultaneously, thus constructing a dynamics which is self-similar with respect to $R_A$ for all $A$. We call such dynamics \textbf{maximally self-similar}.

A dynamics cannot be self-similar with respect to a unique transformation $R_A$. If a dynamics is self-similar with respect to some transformation $R_A$ then it is also self-similar with respect to $(R_A)^n$, for any positive integer $n$. But $(R_A)^n$ is itself a similarity transformation: $(R_A)^n=R_{A^n}$, where we define $A^n$ to be the order of cardinality $n|A|$ which is a stack of $n$ copies of $A$ separated by breaks. Therefore there exists no dynamics which is self-similar with respect to a unique transformation. We say that a dynamics $\{\mathbb{P}\}$ is \textbf{minimally self-similar} if there exists a unique order $A$ such that $\{\mathbb{P}\}$ is only self-similar with respect to $(R_A)^n$ for all $n$.

Some self-similar dynamics are neither minimally nor maximally self-similar. Consider a (finite or infinite) collection $\{\widehat{A}\},\{\widehat{B}\},...$ of covtree nodes. Does there exist a dynamics which is only self-similar with repsect to $R_A, R_B,...$ and their respective powers? If every pair of nodes are unrelated in covtree then the answer is yes. On the other hand, suppose that $\{\widehat{A}\}\prec \{\widehat{B}\}$. Then $B$ contains an $A$-break, and let us denote the future of the break by $D$. Then $\mathcal{G}_B=\mathcal{G}_A\mathcal{G}_D$, and a dynamics is self-similar with respect to $R_A$ and $R_B$ if and only if it is self-similar also with respect to $R_D$.

\paragraph{The post condition:}
So far we have discussed the renormalisation associated with breaks. A similar story applies to posts. One conditions on an $\lc{A}$-post in the labeled case, or an $A$-post in the covariant case. Since a post condition is a special case of a break condition, the past is fixed and the future remains dynamical, allowing for a description in terms of an effective dynamics. In the labeled case, the $\lc{A}$-post condition is equivalent to the $\widehat{\lc{A}}$-break condition. Therefore the corresponding renormalisation transformation is obtained from transformation \ref{eq_181201} via the replacements $r\rightarrow 1$ (since $\widehat{\lc{A}}$ has a single maximal element) and $m\rightarrow|\widehat{\lc{A}}|=|\lc{A}|+1$, \textit{i.e.} $Q_{|\widehat{\lc{A}}|,1}$. In the covariant case, the $A$-post condition is equivalent to the $\widehat{A}$-break condition, and the effective dynamics is given by transformation $R_{\widehat{A}}$, obtained from transformation $\ref{eq2507201}$ via $A\rightarrow\widehat{A}$. 

There is an alternative formulation of the effective dynamics after a post \cite{Martin:2000js}. In this alternative formulation, the post is considered a part of the future rather than the past. The future is therefore constrained to have a unique minimal element (the post itself) but is otherwise dynamical. The effective dynamics is not a CSG model. Instead it is an \textit{originary} CSG model, ensuring that every new element is born above the post. The coupling constants renormalise as:
\begin{equation}\begin{split}\label{eq_171201}
&S_m:\{t_k\}\rightarrow\{t^{(m)}_k\}\\
&t^{(m)}_0=0\\
&t^{(m)}_k=\sum_{l=0}^m\binom{m}{l}t_{k+l}  \ \forall \ k>0,\end{split}\end{equation}
where $m=|\lc{A}|$. The renormalisation transformation depends on $\lc{A}$ only via its cardinality, as signified by the label $m$ on the transformation, $S_m$, and on the renormalised couplings, $t_k^{(m)}$. A derivation of transformation \ref{eq_171201} can be found in \cite{Martin:2000js}. The stationary points of $S_m$, for any $m$, are the Originary Transitive Percolation (OTP) models:
\begin{equation}\begin{split}
&t_0=0\\
&t_k=t^k \ \forall \ k>0, \end{split}\end{equation}
where $t$ is any positive real number.

An originary formulation exists also in the covariant case. We say that a covtree dynamics $\{\mathbb{P}\}$ is originary if $\mathbb{P}(\Gamma_1\rightarrow$\{\twoch\}$)=1$. In the originary viewpoint of the $A$-post condition, a (generic) covtree dynamics $\{\mathbb{P}\}$ is mapped onto an \textit{originary} covtree dynamics $\{\mathbb{P}'_A\}$ via the transformation\footnote{The apostrophe on the transition probabilities $\{\mathbb{P}'_A\}$ is used to distinguish between the images of $\{\mathbb{P}\}$ under $R_A$ and $T_A$.}:
\begin{equation}\label{eq2407201}
\begin{split}
&T_A:\{\mathbb{P}\}\rightarrow\{\mathbb{P}'_A\}\\
&\mathbb{P}'_A(\Gamma_1\rightarrow\text{ \{\twoch\} })=1\\
&\mathbb{P}'_A(\Gamma_n\rightarrow\Gamma_{n+1})=\mathbb{P}(\mathcal{G}_A(\Gamma_n)\rightarrow\mathcal{G}_A(\Gamma_{n+1})) \ \forall \  \Gamma_n \succeq \text{ \{\twoch\} }\\
&\mathbb{P}'_A(\Gamma_n\rightarrow \Gamma_{n+1})=0 \ \text{otherwise.}
\end{split}
\end{equation} 

\paragraph{Covariant cosmology and physical dynamics:}\label{par_1} We have seen that the covariant counterpart of cosmic renormalisation takes the form of a family of transformations, each of which maps a given covtree dynamics into another. A summary of the covariant and the labeled transformations is shown in table \ref{table1}.

While our success in adapting the cosmic renormalisation to the covtree framework bodes well for a covariant causal set cosmology, our results will remain purely formal until we are able to identify a class of covtree dynamics to work with. Having said that, these results could be used to advance the search for physical covtree dynamics. In the remainder of this section, we present directions for further study.

The cosmic transformations can be used to study the relationship between covtree dynamics and CSG dynamics:
\begin{itemize}
\item[1.] We have seen that the labeled break transformation $Q_{m,r}$ (\textit{i.e.} transformation~\ref{eq_181201}) depends on the past of the break only via its cardinality and number of maximal elements. Is the condition on a covtree dynamics $\{\mathbb{P}\}$ that $\{\mathbb{P}_A\}=\{\mathbb{P}_B\}$ if and only if $A$ and $B$ have the same cardinality and number of maximal elements necessary for $\{\mathbb{P}\}$ to be a CSG dynamics? Is it sufficient?

\item[2.] The action of $Q_{m,r}$ on the couplings $t_k$ with $k>0$ can be factorised as $Q_{m,r}=M^m(t_k)$, where $M(t_k)=t_k+t_{k+1}$. Does this property bear any relation to the constraint on a covtree dynamics $\{\mathbb{P}\}$ that, for any finite order $A$, the renormalisation transformation can be factorised as $R_A=R^{|A|}$ for some transformation $R$?
\end{itemize}

A comparison of stationary points (in the labeled case) and self-similar dynamics (in the covariant case) also has the potential to shed light on the covariant form of CSG models:
\begin{enumerate}
\item[3.] If a CSG model $\{t_k\}$ is a stationary point of a labeled transformation, is its corresponding covtree dynamics self-similar?

\item[4.] We have seen that a CSG dynamics is a stationary point of both $Q_{m,r}$ and $Q_{n,s}$ if and only if $r=s$. Is the condition on a covtree dynamics $\{\mathbb{P}\}$ that $\{\mathbb{P}\}=\{\mathbb{P}_A\}=\{\mathbb{P}_B\}$ only if $A$ and $B$ have the same number $r$ of maximal elements necessary for $\{\mathbb{P}\}$ to be a CSG dynamics? Is it sufficient? Such a dynamics is neither maximally nor minimally self-similar, and it follows from our previous analysis that for any $r>0$ there exists a family of self-similar dynamics which satisfy this condition. Are these CSG dynamics, or is the relationship between the labeled and covariant formulations more complex?
\end{enumerate}

Covtree dynamics which flow to a self-similar dynamics (cf. family $(c)$) are also of interest:
\begin{enumerate}
\item[5.]  Let us consider the post condition in the originary formulation. A covtree dynamics $\{\mathbb{P}\}$ flows to a self-similar dynamics under $T_A$ if
\begin{equation}\label{eq_14121907}(T_A)^k[\mathbb{P}(\Gamma_n\rightarrow\Gamma_{n+1})]\rightarrow (T_A)^{k+1}[\mathbb{P}(\Gamma_n\rightarrow\Gamma_{n+1})] \ \text{ as } \  k\rightarrow\infty.\end{equation}
If $\{\mathbb{P}\}$ arrives at a self-similar dynamics after $N$ applications of $T_A$, expression~\ref{eq_14121907} simplifies to:
\begin{equation}\label{eq_29072001}\begin{split}&(T_A)^k[\mathbb{P}(\Gamma_n\rightarrow\Gamma_{n+1})]=(T_A)^{k+1}[\mathbb{P}(\Gamma_n\rightarrow\Gamma_{n+1})] \ \forall \  k\geq N\\
\implies &\mathbb{P}(\mathcal{G}_A^{k}(\Gamma_n)\rightarrow\mathcal{G}_A^{k}(\Gamma_{n+1}))=\mathbb{P}(\mathcal{G}_A^{k+1}(\Gamma_n)\rightarrow\mathcal{G}_A^{k+1}(\Gamma_{n+1})) \ \forall \  k\geq N.
\end{split} \end{equation}

The $N=1$ case is of special interest to us. The Transitive Percolation (TP) models are a 1-parameter family of CSG models, defined by $t_0=1$, $t_k=t^k$, $t\in\mathbb{R}^+$. It is easy to show that, under an application of $S_m$, a TP model with parameter $t$ is maped onto the OTP model with the same $t$ value. Does this mean that for a covtree dynamics $\{\mathbb{P}\}$ to be a TP model it must satisfy condition \ref{eq_29072001} with $N=1$?
\end{enumerate}

Finally, covtree is an opportunity to uncover new dynamics with physical features such as:
\begin{enumerate}
\item[6.] Infinitely many breaks or posts: it is known that OTP gives rise to infinitely many posts with unit probability \cite{Alon:1994}. Is this property related to that the fact that it is a stationary point? Do self-similar covtree dynamics give rise to an infinite sequence of posts or breaks? Could this be a feature of the maximally self-similar dynamics?

\item[7.] Causality: when the transformation $R_A$ factorises as ${R_A=R^{|A|}}$, the effective dynamics is independent of the causal structure of the past. Therefore, could the condition that $R_A$ factorises be interpreted as a causality condition on covtree dynamics?
\end{enumerate}

\begin{table}[h!]
\centering
\begin{tabular}{|c|| c| c |c|} 
 \hline
& \textbf{break} & \multicolumn{2}{|c|}{\textbf{post}}  \\
 &   & \textit{non-originary} & \textit{originary} \\
 \hline \hline
 \textbf{labeled}   & $Q_{m,r}$    &$Q_{m+1,1}$& $S_m$\\
  \hline
 \textbf{covariant}&   $R_A$  & $R_{\widehat{A}}$  &$T_A$\\
  \hline
\end{tabular}
\caption{Summary of transformations. The first row lists the renormalisation transformations which act on CSG couplings. $Q_{m,r}$ is the $\lc{A}$-break transformation, where $m=|\lc{A}|$ and $r$ is the number of maximal elements of $\lc{A}$. $Q_{m+1,1}$ and $S_m$ are the $\lc{A}$-post transformations in the non-originary and originary formulations, respectively. The second row lists the similarity transformations which act on covtree transition probabilities. $R_A$ is the $A$-break transformation. $R_{\widehat{A}}$ and $T_A$ are the $A$-post transformations in the non-originary and originary formulations, respectively.}
\label{table1}
\end{table}

\section{Further structure of covtree}

A pair of challenges on the path to physical covtree dynamics are understanding the structure of covtree and understanding the relationship between paths and their certificates. In this section we present further properties of covtree and its certificates, and illustrate with a toy example how an understanding of the structure of covtree could be a useful tool for constraining covtree dynamics.

\subsection{Nodes}
In this section we list properties which pertain to nodes, including criteria for a set of $n$-orders to be a node, properties of minimal certificates and a study of direct descendants and valency.

We start with the simple property:
\begin{property}
For any finite order $C$ there is a singleton node $\{C\}$ in covtree, and $C$ is the unique minimal certificate of $\{C\}$.
\end{property}

Recall that $\Gamma_{m}\succ\Gamma_n$ in covtree if and only if every certificate of $\Gamma_m$ is a certificate of $\Gamma_n$. Therefore $\{C\}\succ\Gamma_n$ in covtree if and only if $C$ is a certificate of $\Gamma_n$. Since every node in covtree has countably many finite certificates, we find that:
\begin{property} Every node in covtree has countably many singleton descendants. \end{property}

Next we note that, since $C$ is the only $|C|$-stem in the covering order $\widehat{C}$, the node $\{\widehat{C}\}$ is directly above $\{C\}$ in covtree and therefore:
\begin{property} \label{hat}
Every singleton has at least one direct descendant which is a singleton.
\end{property}

Moreover,
\begin{property} \label{unique_sing_above_sing_property}
If $\{\widehat{C}\}$ is the only singleton directly above $\{C\}$ then $\{\widehat{C}\}$ is the only node directly above $\{C\}$.\end{property}
To see this, assume for contradiction that there exists some node $\{A^1, ...,A^k\}$ which is directly above $\{C\}$, where $k\geq 2$ . It follows from the definition of covtree that $C$ is the unique $|C|$-stem in every $A^i\in\{A^1, ...,A^k\}$ and therefore each of the nodes $\{A^1\}, ..., \{A^k\}$ are singletons directly above $\{C\}$, which is a contradiction.

Every singleton with valency greater than one has at least one direct descendant which is a doublet since:
\begin{property}\label{prop_22060201}
If $\{D_{n+1}\}\succ\{C_n\}$ and $D_{n+1}\not=\widehat{C_n}$ then $\{\widehat{C_n},D_{n+1}\}\succ \{C_n\}$.
\end{property}

A corollary of properties \ref{unique_sing_above_sing_property} and \ref{prop_22060201} is:
\begin{property}
No singleton has a valency of 2.
\end{property}

Singletons which possess property \ref{unique_sing_above_sing_property} are the only nodes in covtree which have exactly one direct descendant since: 
\begin{property} \label{valency}
Only singletons can have exactly one direct descendant in covtree.
\end{property}
To prove this, consider the node $\Gamma_n=\{A^1,...,A^k\}$ with $k\geq2$, and assume for contradiction that $\Gamma_{n+1}$ is the only direct descendant of $\Gamma_{n}$.

First, we show that $\Gamma_{n+1}$ contains the covering order $\widehat{A^i}$ of every $A^i\in \Gamma_n$. Suppose for contradiction that $\Gamma_{n+1}$ does not contain the covering order $\widehat{A^i}$ of some $A^i \in \Gamma_n$. Let $C_m$ be an $m$-order that is a certificate of $\Gamma_{n+1}$. Then there exists an $(m+1)$-order $D_{m+1}$ which contains $C_m$ and $\widehat{A^i}$ as stems (a representative of $D_{m+1}$ can be constructed by taking a representative of $C_m$ and adding an element above the stem which is isomorphic to a representative of $A^i$). $D_{m+1}$ is a certificate of $\Gamma'_{n+1}=\Gamma_{n+1}\cup \{\widehat{A^i}\}$, and $\Gamma'_{n+1}\succ\Gamma_n$. This contradicts the assumption that $\Gamma_{n+1}$ is the only direct descendant of $\Gamma_{n}$ . Hence $\Gamma_{n+1}$ contains the covering order $\widehat{A^i}$ for all $A^i\in \Gamma_n$.

Now we show that, since $\Gamma_{n+1}$ contains the covering order $\widehat{A^i}$ of some $A^i\in \Gamma_n$, it cannot be the only direct descendant of $\Gamma_{n}$. Let the $p$-order $E_p$ be a minimal certificate of $\Gamma_n$. Therefore $\Gamma_n\prec \Gamma_{n+1}\prec\{E_p\}$ and hence $E_p$ is a certificate of $\Gamma_{n+1}$. Then there exists a $(p-1)$-order $F_{p-1}$ such that $F_{p-1}$ is a stem in $E_p$ \textit{and} $\widehat{A^i}$ is not a stem in $F_{p-1}$  (a causet isomorphic to a representative of $F_{p-1}$ can be constructed by taking a representative of $E_p$ and removing the element which is maximal in the stem isomorphic to a representative of $A^i$). Then $F_{p-1}$ is a certificate of $\Gamma_n$, which contradicts the assumption that $E_p$ is a minimal certificate of $\Gamma_n$. Therefore, only singletons can have exactly one direct descendant in covtree.

Additionally, 
\begin{property}\label{prop_crown} For any $k\geq 1$ there is a singleton $\{C\}$ in covtree with $k$ singletons directly above it. \end{property}
An immediate corollary is that the valency of singletons is unbounded. (Note that $k$ is not the valency of $\{C\}$, for if $k\geq 2$ then $\{C\}$ has additional direct descendants which are not singletons, cf. property \ref{prop_22060201}.)

An example of a singleton node with 1 singleton directly above it is $\Gamma_4=\{$\twoch \twoch $\}$. 
To show that the statement is true for $k>1$, we construct a countable sequence of singletons $\{C_{n_2}\},\{C_{n_3}\},...,\{C_{n_k}\},...$ such that $\{C_{n_k}\}$ has $k$ singletons directly above it. Figure \ref{singdesc_new} shows the first three singletons in the sequence and their respective singleton descendants.

\begin{figure}[p]
    \includegraphics[width=.85\textwidth]{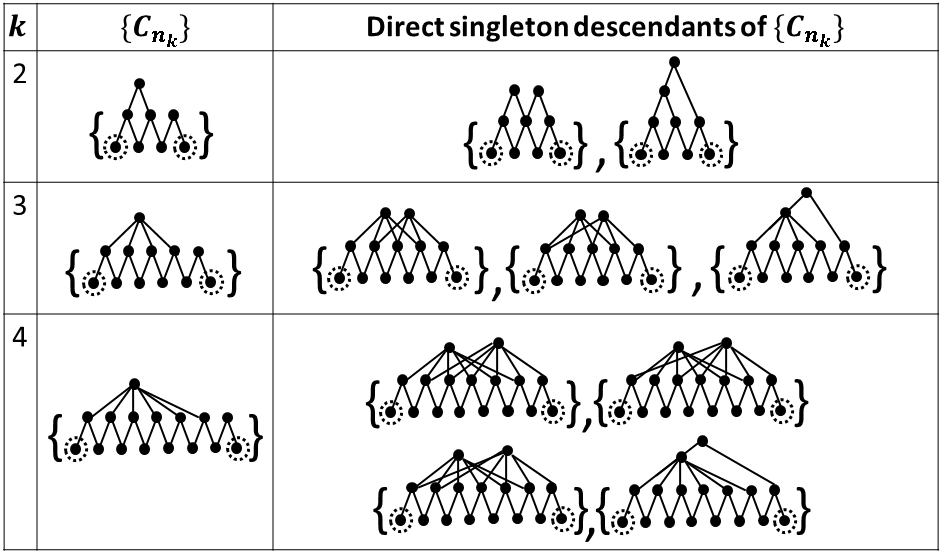}
  \caption{Illustration of property \ref{prop_crown}. The elements circled by a dotted line are identified with each other.}\label{singdesc_new}
\end{figure}
\begin{figure}[p]
    \includegraphics[width=.9\textwidth]{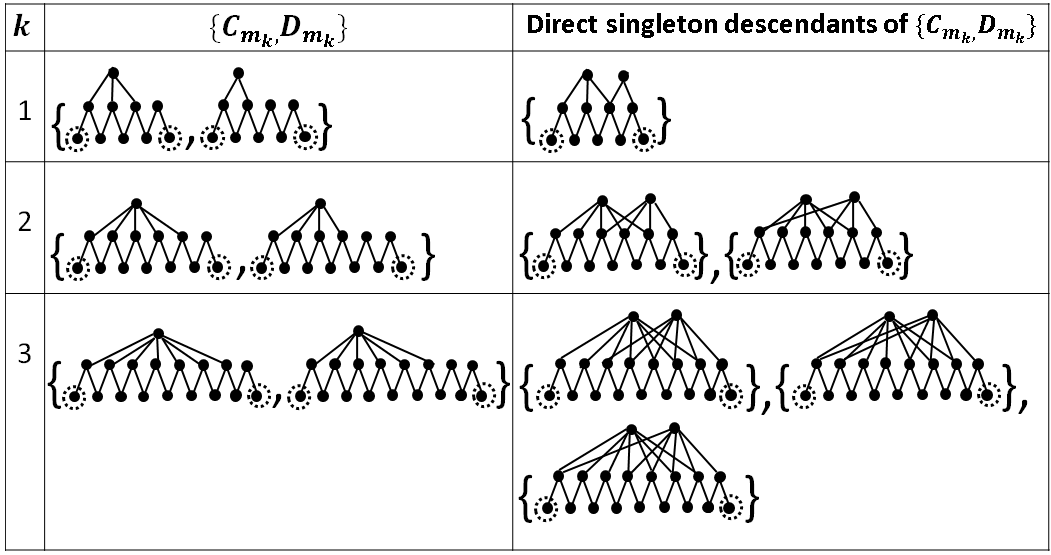}
   \caption{Illustration of property \ref{prop_200801}. The elements circled by a dotted line are identified with each other.}\label{fig_060904b}
\end{figure}

Let us explain the pattern shown in figure \ref{singdesc_new}. Each order $C_{n_k}$ has cardinality $n_k=4k-1$. A representative $\lc{C}_{n_k}$ of $C_{n_k}$, contains $\frac{n_k-1}{2}$ elements in level 1, $\frac{n_k-1}{2}$ elements in level 2, and a single element in level 3. Each element in level 1 is below two elements in level 2. Each element in level 2 is above two elements in level 1. The element in level 3 is above all but one of the elements in level 2.

We construct the singleton descendants of $\{C_{n_k}\}$ as follows. Consruct a new causet from $\lc{C}_{n_k}$ by adding a new element directly above all but one of the level 2 elements subject to the constraint that no level 2 element is maximal in the resulting causet. There are $2k-2$ ways to do this, leading to a collection of $2k-2$ causets. One can show that each of these causets is order-isomorphic to exactly one other in the collection, and taking the corresponding orders gives $k-1$ distinct $(n_k+1)$-orders whose only $n_k$-stem is $C_{n_k}$. A singleton containing each of these orders is directly above $\{C_{n_k}\}$. The additional singleton descendant is $\{\widehat{C_{n_k}}\}$.

Similarly,
\begin{property}\label{prop_200801} For any integer $k\geq 1$ there exists a doublet in covtree with $k$ singletons directly above it. \end{property}

One can construct an infinite sequence of doublets, $\{C_{m_1},D_{m_1}\},\{C_{m_2},D_{m_2}\},...,$ $\{C_{m_k},D_{m_k}\},...$, such that the $k^{th}$ doublet in the sequence has $k$ singletons directly above it. Figure \ref{fig_060904b} shows the first three doublets in the sequence and their direct singleton descendants.

A representative $\lc{C}_{m_k}$ of $C_{m_k}$ has cardinality $m_k=4k+5$ and partial ordering as described below property \ref{prop_crown} (with $n_k$ replaced by $m_k$). A representative of $D_{m_k}$ has these same properties, except that the element in level 3 is above all but \textit{two} of the elements in level 2. The two elements missed out must have a common element in their pasts.

We construct the singleton descendants of $\{C_{m_k},D_{m_k}\}$ as follows. Consruct a new causet from $\lc{C}_{m_k}$ by adding a new element directly above all but two of the level 2 elements subject to the constraints that no level 2 element is maximal in the resulting causet \textit{and} that the two elements which are missed out have a common element in their pasts. In this way, one generates a collection of $2k$ causets. One can show that each of these causets is isomorphic to exactly one other in the collection, and taking the corresponding orders gives $k$ distinct $(m_k+1)$-orders whose set of $m_k$-stems is $\{C_{m_k},D_{m_k}\}$. This proves property \ref{prop_200801}.

A key hurdle in the construction of covtree is understanding which sets of $n$-orders are covtree nodes. The following property gives a necessary condition in the case of doublets:
\begin{property}\label{prep_1307201}
$\{A_n,B_n\}$ is a doublet in covtree only if there exists an $(n-1)$-order $S$ which is a stem in both $A_n$ and $B_n$.
\end{property}

To prove property \ref{prep_1307201}, let $\lc{E}$ be a labeled minimal certificate of $\{A_n,B_n\}$. Let $\lc{A}_n$ and $\lc{B}_n$ be stems in $\lc{E}$ which are isomorphic to representatives of $A_n$ and $B_n$, respectively. Define $\lc{S}:=\lc{A}_n\cap \lc{B}_n$. We will show that $|\lc{S}|=n-1$ and property \ref{prep_1307201} follows.

Note that $\lc{E}=\lc{A}_n \cup \lc{B}_n$, for otherwise $\lc{E}$ would not be minimal. Define $k:=n-|\lc{S}|$ and suppose for contradiction that $1<k\leq n$. Let $\lc{A}_n\setminus \lc{B}_n=\{v_1,v_2,...,v_k\}$ and $\lc{B}_n\setminus \lc{A}_n=\{w_1,w_2,...,w_k\}$. Without loss of generality, let $\lc{F}=\lc{S} \cup \{v_1, v_2,..., v_{k-1}, w_1\}$ be an $n$-stem in $\lc{E}$. Then either $\lc{F}\cong \lc{A}_n$ or $\lc{F}\cong \lc{B}_n$. Suppose $\lc{F}\cong \lc{A}_n$. Then $\lc{E}\setminus \{v_k\}$ is isomorphic to a labeled certificate of $\{A_n,B_n\}$ and is a stem in $\lc{E}$, which contradicts the assumption that $\lc{E}$ is minimal. Suppose $\lc{F}\cong \lc{B}_n$. Then $\lc{E}\setminus \{w_2, w_3, ..., w_k\}$ is isomorphic to a labeled a certificate of $\{A_n,B_n\}$ and is a stem in $\lc{E}$, which is again a contradiction. Hence $|\lc{S}|=n-1$, which completes the proof.

Property \ref{prop_270801} is a corollary:
\begin{property}\label{prop_270801}
If $\Gamma_n$ is a doublet in covtree then all minimal certificates of $\Gamma_n$ are $(n+1)$-orders.
\end{property}
Therefore, if $\Gamma_n$ is a doublet in covtree and $\Gamma_n\prec\Gamma_{n+1}$ then $\Gamma_{n+1}$ contains some minimal certificate of $\Gamma_n$. It is a corollary of properties \ref{prop_200801} and \ref{prop_270801} that for any integer $k\geq 1$ there exists a doublet in covtree with $k$ minimal certificates.

\subsection{Paths}
In this section, we present properties of certain covtree paths and their certificates.
\begin{property} In covtree, there are infinite upward-going paths from the origin in which every node is a singleton.\end{property}
We call the subset of covtree which contains exactly all these paths \textbf{singtree} -- a tree of singletons. Figure \ref{singtree} shows the first three levels of singtree.

\begin{figure}[h]
  \centering
	\includegraphics[width=0.6\textwidth]{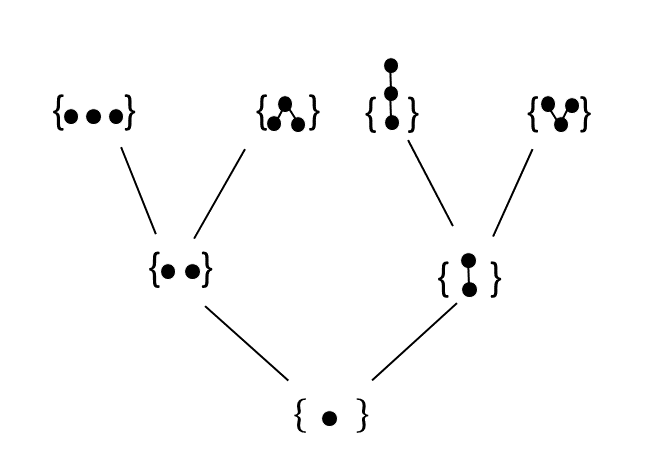}
	\caption{The first three levels of singtree.}
	\label{singtree}
\end{figure}

To discuss singtree we will need the concept of the Newtonian order. A \textbf{Newtonian causet} is a causet in which every element in level $k$ is above every element in level $k-1$. A \textbf{Newtonian order} is an order whose representatives are Newtonian.

In a Newtonian causet, every pair of elements which are unrelated have the same past and the same future, alluding to a notion of a Newtonian global time, hence its name\footnote{A Newtonian order is \textit{not} a good approximation of continuum Euclidean space.}. A Newtonian causet is a ``stack of antichains'', and for any natural number $N$, the union of the first $N$ levels is a past of a break. The local finiteness condition implies that every level whose elements are not maximal must be finite.

To characterise the nodes of singtree and the certificates of singtree paths we will need the following lemma about Newtonian orders:
\begin{lemma}\label{lemma_050902} The following properties of a finite or infinite order, $C$, are equivalent:
\begin{itemize}
\item[(i)] $C$ has a unique representative,
\item[(ii)] for every natural number $n\leq |C|$ there is a unique $n$-order which is a stem in $C$,
\item[(iii)] $C$ is Newtonian.
\end{itemize}
\end{lemma}

\begin{proof}
Let $\lc{C}$ denote a representative of $C$. Recall that $L(x)$ denotes the level of element $x$.
\begin{itemize}
\item[$(i) \implies (ii)$] Suppose for contradiction that $C$ has two $n$-stems, $C_n\not=C'_n$, for some $n \leq |C|$. Then there is a representative of $C$ whose restriction to the interval $[n-1]=\{0,1,...,n-1\}$ is a representative of $C_n$, and similarly for $C'_n$. Hence $C$ has at least two representatives. Contradiction.

\item[$(ii) \implies (iii)$] Suppose for contradiction that $C$ is not Newtonian, \textit{i.e.} there exist some $x, y\in\lc{C}$ such that $L(x)>L(y)$ and $x\not\succ y$. Let $\tilde{S}$ be the union of the inclusive past of $x$ with levels $1,2,...,L(y)$. Note that $\tilde{S}\subset\lc{C}$ is a stem in $\lc{C}$. Then $\tilde{S}\setminus\{x\}$ and $\tilde{S}\setminus\{y\}$ are non-isomorphic stems in $\lc{C}$. (To see that they are non-isomorphic, note that $\tilde{S}\setminus\{x\}$ contains $L(x)-1$ levels and $\tilde{S}\setminus\{y\}$ contains $L(x)$ levels.) Hence $C$ has at least two $n$-orders as stems, where $n=|\tilde{S}|-1$. Contradiction.

\item[$(iii) \implies (i)$] Let $f:\lc{C}\rightarrow\lc{C}'$ be an order-isomorphism between two representatives of $C$. The Newtonian condition restricts the action of $f$ on each level in $\lc{C}$ to be a permutation, and therefore $f$ must be the identity map. Hence $C$ has a unique representative.
\end{itemize}\end{proof}

The equivalence of statements $(ii)$ and $(iii)$ in lemma \ref{lemma_050902} implies that:
\begin{property}A singleton $\{C_n\}$ is in singtree if and only if $C_n$ is Newtonian.\end{property}
If $\{C_n\}$ is a node in singtree then it has exactly two direct descendants in singtree: $\{\widehat{C_n}\}$ and $\{D_{n+1}\}$, where $\tilde{D}_{n+1}$ is the Newtonian order whose representative is constructed from a representative of $C_n$ by adding a new element to its maximal level. If $\{C_n\}$ is a node in singtree then it has exactly three direct descendants in covtree: its singtree descendants, $\{\widehat{C_n}\}$ and $\{D_{n+1}\}$, and the doublet $\{\widehat{C_n},D_{n+1}\}$.

A second corollary of lemma \ref{lemma_050902} is:
\begin{property}\label{corollary_170901} An infinite order $C$ is Newtonian if and only if it is a certificate of a singtree path. \end{property}

Given property \ref{corollary_170901}, it is now a simple matter to solve for the family of covtree dynamics in which the set of non-Newtonian orders is null: it is the set of covtree walks in which the walker stays in singtree with probability 1, \textit{i.e.}
\begin{equation}
\mathbb{P}(\Gamma_n)=0 \ \forall \ \Gamma_n \text{ not in singtree.} 
\end{equation}
While this family of dynamics is not physically interesting, it acts as a proof of principle, illustrating how an understanding of covtree could allow one to solve for a dynamics with particular features.

\section{Discussion} \label{secdiscussion}
Taking our cue from \cite{Dowker:2019qiz}, in this work we studied the structure of covtree and its certificates as a first step to constructing physically well-motivated covtree dynamics. 

We made progress on the cosmological front. We identified covtree's self-similar structure as well as which nodes correspond to posts and breaks. This allowed us to write down a similarity transformation between covtree dynamics, akin to the cosmic renormalisation of CSG models. As well as being a starting point for a paradigm of covariant causal set cosmology, these developments provide us with a concrete arena in which to investigate which covtree dynamics are the physically interesting ones, \textit{e.g.} by allowing us to propose a causality condition on covtree transition probabilities.

We also presented a glimpse of covtree's intricate structure. For example, we saw that there is no upper bound on the valency of covtree nodes, even in the simplest case of singletons. After a study of singtree and its certificates, we solved for the class of dynamics in which the set of Newtonian orders has measure 1, thereby providing a toy example of how an understanding of the structure of covtree could be utilised in writing down dynamics with desired features.  But, since these Newtonian dynamics are not physically interesting, this is very much a case of ``looking under the lamp-post''.

Where are we to look if not under the lamp-post? One avenue for exploration is to ask: what role, if any, do \textit{rogues} play in the physics of covtree walks?

Since in CSG models the set of rogues is null \cite{Brightwell:2002vw}, identifying covtree dynamics which possess this property is a step towards understanding what form CSG dynamics take on covtree. Moreover, if following \cite{Brightwell:2002vw} we are to choose $\mathcal{R}$ to be our sigma-algebra of observables then --- unless the covtree measure on $\mathcal{R}(\mathcal{S})$ has a unique extension to $\mathcal{R}$ --- one is faced with ambiguities both in interpretation and calculation. It is sufficient that the set of rogues be null for there to exist a unique extention, and therefore rogue-free dynamics are compatible with this approach. Finally, rogue-free dynamics are cosmologically interesting. A rogue causet contains an infinite level and as a result cannot contain an infinite sequence of posts or breaks. Therefore, that the dynamics is rogue-free is a necessary condition for the dynamics to be relevant for our cosmological paradigm.

One can draw an analogy between the condition that the set of rogues is null and the condition that the set of non-Newtonian orders is null: the former is the condition that the set of paths with more than one certificate is null, the latter the condition that the set of paths with more than one \textit{labeled} certificate is null. 
However, while we were able to solve for the latter, solving for the former poses a new challenge because it is a limiting condition: at no finite stage of the covtree walk can the claim that the growing order is a rogue be verfied or falsified. This is because for every node in covtree there exist both an infinite certificate which is a rogue and an infinite certificate which is not a rogue.

This means that there is no rogue analogue to singtree. Instead, we must look for other ways to obtain rogue-free dynamics. Pursuing the strictly stronger condition that the dynamics gives rise to infinitely many posts or breaks with unit probability is a promising route. We already know the defining feature of these covtree walks: that, with unit probability, they pass through infinitely many nodes of the form $\{\widehat{A}\}$. The challenge ahead is to formulate this feature in terms of covtree transition probabilities.

\par \textbf{Acknowledgments:} Parts of this work are a result of collaboration with Fay Dowker, Amelia Owens and Nazireen Imambaccus and have previously appeared in \cite{Imambaccus:2018, Owens:2018}. The author thanks the Perimeter Institute and the Raman Research Institute for hospitality while this work was being completed. The author is partially supported by the Beit Fellowship for Scientific Research and by the Kenneth Lindsay Scholarship Trust.

\appendix
\newpage

\section{Table of sets defined in the text}
\FloatBarrier
\begin{table}[htpb!]\centering
\begin{tabular}{| l | l |}
\hline
  $\tilde{\Omega}(\mathbb{N})$ & The set of finite labeled causets \\ \hline
  $\tilde{\Omega}$ & The set of infinite labeled causets \\ \hline
  $\Omega(n)$ & The set of $n$-orders for some $n\in\mathbb{N}^+$ \\ \hline
  $\Gamma_n$ &  A subset of $\Omega(n)$\\ \hline
  $\Omega$ & The set of infinite orders \\ \hline
  $cyl(\lc{A})$ &  $cyl(\lc{A})=\{ \lc{C}\in\tilde{\Omega} \mid \lc{A} \text{ is a stem in } \tilde{C}\}$\\ \hline
  $cert(\Gamma_n)$ & $cert(\Gamma_n)=\{ C\in\Omega \mid C \text{ is a certificate of } \Gamma_n\}$ \\ \hline
  $\tilde{\mathcal{R}}$ & The sigma-algebra generated by the cylinder sets \\ \hline
  $\mathcal{R}$ & $\mathcal{R}:=\{\mathcal{E}\in \lc{\mathcal{R}}| \lc{C}\in\mathcal{E} \text{ and }  \lc{C}\cong\lc{D} \implies \lc{D}\in\mathcal{E}\}$ \\ \hline
  $\mathcal{R}(\mathcal{S})$ & The sigma-algebra generated by the certificates sets \\ \hline
  $\Lambda$ &  $\Lambda = \bigcup\limits_{n\in\mathbb{N}^+}\{\Gamma_n \subseteq \Omega(n)| \exists \text{ a certificate for } \Gamma_n\}$\\ \hline
  $\Lambda_A$ &  $\Lambda_A = \{\Gamma_n \in \Lambda| \Gamma_n\succeq  \{\widehat{A}\}\}$ \\\hline
\end{tabular}
\caption{Table of sets defined in the text.}
\end{table}
\FloatBarrier
\bibliography{covtree}{}
\bibliographystyle{unsrt}

\end{document}